\documentclass{article}
\usepackage{amssymb}

\usepackage{etex}
\usepackage{amsmath}
\usepackage{tikz}
\usepackage{tkz-graph}
\usepackage{bussproofs}
\usepackage{graphicx}

\usetikzlibrary{arrows,automata}
\usetikzlibrary{arrows,shapes,positioning}
\usetikzlibrary{tikzmark,decorations.pathreplacing,calc}
\newtheorem{theorem}{Theorem}

\newtheorem{claim}[theorem]{Claim}
\newtheorem{conclusion}[theorem]{Conclusion}

\newtheorem{corollary}[theorem]{Corollary}

\newtheorem{definition}[theorem]{Definition}
\newtheorem{example}[theorem]{Example}

\newtheorem{lemma}[theorem]{Lemma}

\newtheorem{remark}[theorem]{Remark}

\newtheorem{partial solution}[theorem]{Partial Solution}

\newenvironment{proof}[1][Proof]{\textbf{#1.} }{\ \rule{0.5em}{0.5em}}

\input{tcilatex}

\begin{document}

\begin{center}
{\Large Proof Compression and NP Versus PSPACE\smallskip . Part 2 }%
\marginpar{
July 2019}

L. Gordeev

l\texttt{ew.gordeew@uni-tuebingen.de\medskip }
\end{center}

\textbf{Abstract. }{\small We upgrade \cite{GH} to a complete proof of the
conjecture NP = PSPACE} {\small that is known as one of the fundamental open
problems in the mathematical theory of computational complexity. Since
minimal propositional logic is known to be PSPACE complete, while PSPACE to
include NP, it suffices to show that every valid purely implicational
formula }$\rho ${\small \ has a proof whose weight (= total number of
symbols) and time complexity of the provability involved are both polynomial
in the weight of }$\rho ${\small . As is \cite{GH}, we use proof theoretic
approach -- in both sequential and natural deduction forms. Recall that in 
\cite{GH} we considered any valid }$\rho ${\small \ in question that had (by
the definition of validity) a ``short'' tree-like proof }$\pi ${\small \ in
the Hudelmaier-style cutfree sequent calculus for minimal logic. The
``shortness'' means that the height of }$\pi ${\small , the total number and
maximum weight of different formulas occurring in it are all polynomial in
the weight of }$\rho ${\small . However, the size (= total number of nodes),
and hence also the weight, of }$\pi $ {\small could be exponential in that
of }$\rho ${\small . To overcome this trouble we embedded }$\pi $ {\small %
into Prawitz's proof system of natural deductions containing single
formulas, instead of sequents. As in }$\pi ${\small , the height, the total
number and maximum formula weight of the resulting tree-like natural
deduction }$\partial _{1}${\small \ are polynomial, whereas the size of }$%
\partial _{1}$ still {\small could be exponential, in the weight of }$\rho $%
{\small . In our next, crucial move, }$\partial _{1}${\small \ was
deterministically compressed into a ``small'' dag-like deduction }$\partial $%
{\small \ whose horizontal levels contained only mutually different
formulas, which made the whole weight polynomial in that of }$\rho ${\small %
. However, }$\partial ${\small \ required a more complicated verification of
the underlying provability of }$\rho ${\small . In the present paper we
further compress }$\partial ${\small \ into a desired deduction }$\partial
_{0}${\small \ that deterministically proves }$\rho ${\small \ in time and
space polynomial in the weight of }$\rho ${\small . [Working in a natural
deduction calculus is essential because tree-to-dag horizontal compression
of }$\pi ${\small \ merging equal sequents, instead of formulas, is
(possible but) insufficient, since the total number of different sequents
occurring in }$\pi ${\small \ might be exponential in the weight of }$\rho $%
{\small \ -- even assuming that all formulas occurring in sequents are
subformulas of }$\rho $.{\small ]}

\section{Introduction}

Gordeev and Haeusler \cite{GH} presented a dag-like version of Prawitz's 
\cite{Prawitz} tree-like natural deduction calculus for minimal logic, 
\textsc{NM}$_{\rightarrow }$, and left open a problem of computational
complexity of the dag-like provability involved (\cite{GH}: Problem 22). In
this paper we show a solution that proves the conjecture \textbf{NP = PSPACE}%
. To explain it briefly first consider plain tree-like provability. Recall
that our basic deduction calculus \textsc{NM}$_{\rightarrow }$ includes two
standard inferences 
\begin{equation*}
\fbox{$\left( \rightarrow I\right) :\dfrac{\QDATOP{\QDATOP{\left[ \alpha %
\right] }{\vdots }}{\beta }}{\alpha \rightarrow \beta }$}\ ,\quad \fbox{$\
\left( \rightarrow E\right) :\dfrac{\alpha \quad \quad \alpha \rightarrow
\beta }{\beta }$}
\end{equation*}
and one auxiliary repetition rule $\fbox{$\left( \rightarrow R\right) :%
\dfrac{\alpha }{\alpha \ }$}$, where $\left[ \alpha \right] $ in $\left(
\rightarrow I\right) $\ indicates that all $\alpha $-leaves occurring above $%
\beta $-node exposed are \emph{discharged} assumptions.

\begin{definition}
\begin{enumerate}
\item  A given (whether tree- or dag-like) \textsc{NM}$_{\rightarrow }$%
-deduction $\partial $ \emph{proves} its root-formula $\rho $ (abbr.: $%
\partial \vdash \rho $) iff every maximal thread connecting the root with a
leaf labeled $\alpha $ is closed (= discharged), i.e. it contains a $\left(
\rightarrow I\right) $ with conclusion $\alpha \rightarrow \beta $, for some 
$\beta $.

\item  A purely implicational formula $\rho $ is \emph{valid in minimal logic%
} iff there exists a tree-like \textsc{NM}$_{\rightarrow }$-deduction $%
\partial $ that proves $\rho $; such $\partial $ is called a \emph{proof} of 
$\rho $.
\end{enumerate}
\end{definition}

\begin{remark}
Tree-like constraint in the second part of definition is inessential. That
is, for any dag-like $\partial \in \,$\textsc{NM}$_{\rightarrow }$ with
root-formula $\rho $, if $\partial \vdash \rho $ then $\rho $ is valid in
minimal logic.
\end{remark}

This is because any given dag-like $\partial $ can be unfolded into a
tree-like deduction $\partial ^{\prime }$ by straightforward
thread-preserving bottom-up recursion. Namely, every node $x\in \partial $
having $n>1$\ distinct conclusions should be replaced by $n$ distinct nodes $%
x_{1},\cdots ,x_{n}\in \partial ^{\prime }$ with corresponding single
conclusions and (identical) premises of $x$. This operation obviously
preserves the closure of threads, i.e. $\partial \vdash \rho $ infers $%
\partial ^{\prime }\vdash \rho $.

Formal verification of the assertion $\partial \vdash \rho $ is simple, as
follows -- whether for tree-like or generally dag-like $\partial $. Every
node $x\in \partial $ is assigned, by descending recursion, a set of
assumptions $A\left( x\right) $ such that:

\begin{enumerate}
\item  $A\left( x\right) :=\left\{ \alpha \right\} $ if $x$ is a leaf
labeled $\alpha $,

\item  $A\left( x\right) :=A\left( y\right) $ if $x$ is the conclusion of $%
\left( \rightarrow R\right) $ with premise $y$,

\item  $A\left( x\right) :=A\left( y\right) \setminus \left\{ \alpha
\right\} $ if $x$ is the conclusion of $\left( \rightarrow I\right) $ with
label $\alpha \rightarrow \beta $\ and premise $y$,

\item  $A\left( x\right) :=A\left( y\right) \cup A\left( z\right) $ if $x$
is the conclusion of $\left( \rightarrow E\right) $ with premises $y,$ $z$.
\end{enumerate}

This easily yields

\begin{lemma}
Let $\partial \in \,$\textsc{NM}$_{\rightarrow }$ (whether tree- or
dag-like). Then $\partial \vdash \rho \Leftrightarrow A\left( r\right)
=\emptyset $ holds with respect to standard set-theoretic interpretations of 
$\cup $ and $\setminus $ in $A\left( r\right) $, where $r$ and $\rho $\ are
the root and root-formula of $\partial $, respectively. Moreover, $A\left(
r\right) \overset{?}{=}\emptyset $ is verifiable by a deterministic TM in $%
\left| \partial \right| $-polynomial time, where by $\left| \partial \right| 
$ we denote the weight (i.e. total number of symbols) of $\partial $. 
\footnote{{\footnotesize The latter is completely analogous to the
well-known polynomial-time decidability of the circuit value problem (see
also Appendix).}}
\end{lemma}

Now let us upgrade \textsc{NM}$_{\rightarrow }$ to \textsc{NM}$_{\rightarrow
}^{\flat }$\ by adding a new \emph{separation} rule $\left( \rightarrow
S\right) $%
\begin{equation*}
\fbox{$\left( \rightarrow S\right) :\dfrac{\overset{n\ times}{\overbrace{%
\alpha \quad \cdots \quad \alpha }}}{\alpha \ }\ $($n$ arbitrary) }
\end{equation*}
whose identical premises are understood disjunctively: ``\emph{if at least
one premise is proved then so is the conclusion}'' (in contrast to ordinary
conjunctive inference: ``\emph{if all premises are proved then so is the
conclusion}''). Note that in dag-like deductions the nodes might have
several conclusions (unlike in tree-like ones). The modified assignment $A$\
in \textsc{NM}$_{\rightarrow }^{\flat }$ (that works in both tree-like and
dag-like cases) is defined by adding to old recursive clauses 1--4 (see
above) a new clause 5 with new separation symbol $\circledS $ :

\begin{description}
\item  5. $A\left( x\right) =\circledS \left( A\left( y_{1}\right) ,\cdots
,\,A\left( y_{n}\right) \right) $ if $x$ is the conclusion of $\left(
\rightarrow S\right) $ with premises $y_{1},\cdots ,y_{n}$.
\end{description}

\begin{claim}
For any dag-like deduction $\partial \in \,$\textsc{NM}$_{\rightarrow
}^{\flat }$\ whose root $r$ is labeled $\rho $, $\rho $ is valid in minimal
logic, provided that $A\left( r\right) $ reduces to $\emptyset $ (abbr.: $%
A\left( r\right) \vartriangleright \emptyset $) by standard set-theoretic
interpretations of $\cup $, $\setminus $ and nondeterministic disjunctive
valuation $\circledS \left( t_{1},\cdots ,\,t_{n}\right) :=t_{i}$, for any
chosen $i\in \left\{ 1,\cdots ,n\right\} $. Moreover, the assertion $A\left(
r\right) \vartriangleright \emptyset $ (that is also referred to as `$%
\partial $ \emph{proves} $\rho $') can be confirmed by a nondeterministic TM
in $\left| \partial \right| $-polynomial time.
\end{claim}

This claim reduces to its trivial \textsc{NM}$_{\rightarrow }$ case (see
above). For suppose that $A\left( r\right) \vartriangleright \emptyset $
holds with respect to successive nondeterministic valuation of the
occurrences $\circledS $. This reduction determines successive ascending
(i.e. bottom-up) thinning of $\partial $ that results in a dag-like $%
\partial _{0}\in \,$\textsc{NM}$_{\rightarrow }^{\flat }$, while $A\left(
r\right) \vartriangleright \emptyset $ in $\partial $\ implies $A\left(
r\right) =\emptyset $ in $\partial _{0}$. Since $\left( \rightarrow S\right) 
$ does not occur in $\partial _{0}$ anymore, we have $\partial _{0}\in \,$%
\textsc{NM}$_{\rightarrow }$. \emph{That is, there is a dag-to-dag operation 
}\textsc{NM}$_{\rightarrow }^{\flat }$\emph{\ }$\ni \partial \hookrightarrow
\partial _{0}\in \,$\textsc{NM}$_{\rightarrow }$\emph{\ that preserves
provability, the root-formula and the weight upper bound}. By previous
considerations with regard to \textsc{NM}$_{\rightarrow }$ we conclude that $%
\rho $ is valid in minimal logic, which can be confirmed\ in $\left|
\partial \right| $-polynomial time, as required.

Since minimal logic is PSPACE complete (\cite{Statman}, \cite{Svejdar}), in
order to arrive at the desired conclusion \textbf{NP = PSPACE }it will
suffice to show that for any valid $\rho $ there is a dag-like deduction $%
\partial \in \,$\textsc{NM}$_{\rightarrow }^{\flat }$\ of $\rho $\
satisfying $A\left( r\right) \vartriangleright \emptyset $, whose size and
maximal formula length are polynomial in the weight $\left| \rho \right| $
of $\rho $. But this is a consequence of \cite{GH}. Recall that in \cite{GH}
we presented a deterministic tree-to-dag \emph{horizontal compression} $%
\left| \rho \right| $-polynomially reducing the weight of a given ``large''
tree-like deduction of $\rho $ in \textsc{NM}$_{\rightarrow }$ that is
obtained by embedding a derivation of $\rho $ in a Hudelmaier-style \cite
{Hud} cutfree sequent calculus. This compression resulted in a suitable
``small'' dag-like \emph{deduction frame} together with a \emph{locally
coherent} set of maximal threads, in the\ multipremise expansion of \textsc{%
NM}$_{\rightarrow }$ (called \textsc{NM}$_{\rightarrow }^{\ast }$). This
pair determines a deduction in \textsc{NM}$_{\rightarrow }^{\flat }$ that
admits a \emph{fundamental set of chains} (see below). In this paper we
describe a nondeterministic dag-to-dag \emph{horizontal cleansing} further
converting such \textsc{NM}$_{\rightarrow }^{\flat }$-deduction\ into a
required ``cleansed'' deduction satisfying $A\left( r\right)
\vartriangleright \emptyset $ (actually $A\left( r\right) =\emptyset $,
already in \textsc{NM}$_{\rightarrow }$).

\subsection{Recollection of \protect\cite{GH}}

Recall that $\rho $ is called \emph{dag-like provable} in \textsc{NM}$%
_{\rightarrow }^{\ast }$ iff there exists a \emph{locally correct deduction
frame} $\widetilde{D}=\!\left\langle D,\text{\textsc{s}},\ell ^{\text{%
\textsc{f}}}\right\rangle $ with root-formula $\rho =\ell ^{\text{\textsc{f}}%
}\left( r\right) $, \footnote{{\footnotesize The weight of }$\widetilde{D}$%
{\footnotesize \ is assumed to be polynomial in that of }$\rho $%
{\footnotesize \ (see \cite{GH}).}} together with a \emph{locally coherent}
pair $\partial \!=\!\left\langle \widetilde{D},G\right\rangle $ such that $G:%
\overrightarrow{\text{\textsc{e}}}\!\left( D\right) \rightarrow \left\{
0,1\right\} $ determines a set of threads that confirms alleged validity of $%
\rho $, where $\overrightarrow{\text{\textsc{e}}}\!\left( D\right) $ denotes
the set of edge-chains in $D$. Such $\partial $ is called a \emph{dag-like\
proof} of $\rho $ in \textsc{NM}$_{\rightarrow }^{\ast }$.\footnote{%
{\footnotesize Here and below basic notions and notations are imported from 
\cite{GH}.}} Note that \textsc{NM}$_{\rightarrow }^{\ast }$-deductions may
have arbitrary many premises and/or conclusions. Without loss of generality
we assume that $\widetilde{D}$ is horizontally compressed, i.e. $\ell ^{%
\text{\textsc{f}}}\!\left( x\right) \neq \ell ^{\text{\textsc{f}}}\!\left(
y\right) $ for all $x\neq y$ on the same level in $D$ (\cite{GH}: Ch. 3.1,
3.2). In \cite{GH} we also observed that the local correctness of $%
\widetilde{D}$ is verifiable in $\left| \rho \right| $-polynomial time,
although local coherence of $\!\left\langle \widetilde{D},G\right\rangle $
has no obvious low-complexity upper bound, as $\overrightarrow{\text{\textsc{%
e}}}\!\left( D\right) $ is generally exponential. The currently proposed
upgrade is based on the \emph{fundamental sets of threads}, instead of $G$
and $\overrightarrow{\text{\textsc{e}}}\!\left( D\right) $, as follows.

\subsection{Upgrade in \textsc{NM}$_{\rightarrow }^{\ast }$}

Let $\widetilde{D}=\!\left\langle D,\text{\textsc{s}},\ell ^{\text{\textsc{f}%
}}\right\rangle $ be a given locally correct deduction frame with
root-formula $\rho =\ell ^{\text{\textsc{f}}}\left( r\right) $, \textsc{K}$%
\left( D\right) $ be the set of maximal ascending chains (also called
threads) consisting of nodes (vertices) $u\in \,$\textsc{v}$\left( D\right) $
connecting root $r$ with leaves. A given set $\mathcal{F}\subset \,$\textsc{K%
}$\left( D\right) $ is a \emph{fundamental set of threads} (abbr.: \emph{fst}%
) in $\widetilde{D}$ if the following three conditions are satisfied, where
for any $\Theta =\left[ r=x_{0},\cdots ,x_{h\left( D\right) }\right] \in \,$%
\textsc{K}$\left( D\right) $ and $i\leq h\left( D\right) $ we let $\Theta
\!\upharpoonright _{x_{i}}:=\left[ x_{0},\cdots ,x_{i}\right] $.

\begin{enumerate}
\item  $\mathcal{F}$ is dense in $D$, i.e. $\left( \forall u\in \text{%
\textsc{v}}\left( D\right) \right) \left( \exists \Theta \in \mathcal{F}%
\right) \left( u\in \Theta \right) $.

\item  Every $\Theta \in \mathcal{F}\,$is closed, i.e. its leaf-formula $%
\ell ^{\text{\textsc{f}}}\!\left( x_{h\left( D\right) }\right) $ is
discharged in $\Theta $.

\item  $\mathcal{F}$ preserves $\left( \rightarrow E\right) $, i.e.

$\left. 
\begin{array}{c}
\left( \forall \Theta \in \mathcal{F}\right) \left( \forall u\in \Theta
\right) \left( \forall v\neq w\in \text{\textsc{v}}\left( D\right)
\!:\left\langle u,v\right\rangle \!,\left\langle u,w\right\rangle \in \text{%
\textsc{e}}\left( D\right) \wedge v\in \Theta \right) \\ 
\left( \exists \Theta ^{\prime }\in \mathcal{F}\right) \left( w\in \Theta
^{\prime }\wedge \Theta \!\upharpoonright _{u}=\Theta ^{\prime
}\!\upharpoonright _{u}\right) \qquad \quad \qquad \qquad \qquad .\qquad
\qquad \quad
\end{array}
\right. $
\end{enumerate}

\begin{lemma}
Let $\widetilde{D}$ be as above and suppose that there exists a \emph{fst} $%
\mathcal{F}$\ in $\widetilde{D}$. Then $\rho $ is dag-like provable in 
\textsc{NM}$_{\rightarrow }^{\ast }$.
\end{lemma}

\begin{proof}
Define $G:\overrightarrow{\text{\textsc{e}}}\!\left( D\right) \rightarrow
\left\{ 0,1\right\} $ by setting $G\left( \overrightarrow{e}\right) =1$ iff $%
\Theta \left[ \overrightarrow{e}\right] \in \mathcal{F}$, where $\Theta %
\left[ \overrightarrow{e}\right] \in \,$\textsc{K}$\left( D\right) $
contains all nodes occurring in the canonical thread-expansion of $%
\overrightarrow{e}$. Then $\partial \!=\left\langle \widetilde{D}%
,G\right\rangle $ is a dag-like\emph{\ }proof of $\rho $. The local
coherence conditions 1, 2, 4, 5 (cf. \cite{GH}: Definition 6) are easily
verified. In particular, 4 follows from the third \emph{fst} condition with
respect to $\mathcal{F}$.
\end{proof}

\begin{lemma}
For any dag-like proof $\left\langle \widetilde{D},G\right\rangle $ of $\rho 
$ there are $D_{0}\subseteq \,D$, $G_{0}:\overrightarrow{\text{\textsc{e}}}%
\!\left( D_{0}\right) \rightarrow \left\{ 0,1\right\} $, $\mathcal{F}\subset
\,$\textsc{K}$\left( D_{0}\right) $\ and a dag-like\emph{\ }proof $%
\left\langle \widetilde{D_{0}},G_{0}\right\rangle $ of $\rho $ such that $%
\mathcal{F}$ is \emph{fst} in $\widetilde{D_{0}}$.
\end{lemma}

\begin{proof}
Let $\mathcal{F}:=\left\{ \Theta :G\left( \overrightarrow{e}\left[ \Theta %
\right] \right) =1\right\} $ for $\overrightarrow{e}\left[ \Theta \right] :=%
\overrightarrow{e_{m}}\in \overrightarrow{\text{\textsc{e}}}\!\left(
D\right) $ determined by $\Theta $ as specified in \cite{GH}: Definition 8.
It is readily seen that such $\mathcal{F}$ is a \emph{fst} in $\widetilde{D}$%
. The crucial condition 3 follows directly from the corresponding local
coherence condition 4 (cf. \cite{GH}: Definition 6). Let $D_{0}\subseteq D$
be the minimum sub-dag containing every edge occurring in $\underset{\Theta
\in \mathcal{F}}{\bigcup }\Theta $ and let $\widetilde{D_{0}}=\!\left\langle
D_{0},\text{\textsc{s}},\ell ^{\text{\textsc{f}}}\right\rangle $ be the
corresponding sub-frame of $\widetilde{D}$. Obviously $\widetilde{D_{0}}$ is
locally correct. Define $G_{0}:\overrightarrow{\text{\textsc{e}}}\!\left(
D_{0}\right) \rightarrow \left\{ 0,1\right\} $ by setting $G\left( 
\overrightarrow{e}\right) =1$ iff $\Theta \left[ \overrightarrow{e}\right]
\in \mathcal{F}$, where $\Theta \left[ \overrightarrow{e}\right] \in \,$%
\textsc{K}$\left( D_{0}\right) $ contains all nodes occurring in the
canonical thread-expansion of $\overrightarrow{e}$. Then $\partial
\!=\left\langle \widetilde{D_{0}},G_{0}\right\rangle $ is a dag-like\emph{\ }%
proof of $\rho $. The crucial density of $\mathcal{F}$ in $D_{0}$ obviously
follows from definitions of $D_{0}$ and $G_{0}$, as every edge in $D_{0}$
occurs in some thread from $\mathcal{F}$, while for any $\overrightarrow{e}%
\in \overrightarrow{\text{\textsc{e}}}\!\left( D_{0}\right) $ we have $%
G_{0}\left( \overrightarrow{e}\right) =1$ iff $\Theta \left[ \overrightarrow{%
e}\right] \in \mathcal{F}$.
\end{proof}

Together with \cite{GH}: Corollaries 15, 20 these Lemmata yield

\begin{corollary}
Any given $\rho $ is valid in minimal logic iff there exists a pair $%
\left\langle \widetilde{D},\mathcal{F}\right\rangle $ such that $\widetilde{D%
}$ is a locally correct deduction frame with root-formula $\rho =\ell ^{%
\text{\textsc{f}}}\left( r\right) $ and $\mathcal{F}$ being a \emph{fst}\ in 
$\widetilde{D}$. We can just as well assume that $\widetilde{D}$ is
horizontally compressed and its weight is polynomial in that of $\rho $.
\end{corollary}

\begin{remark}
We can't afford $\mathcal{F}$ to be by any means polynomial in $\rho $.
However, mere existence of $\mathcal{F}$ enables a nondeterministic polytime
verification of $A\left( r\right) \vartriangleright \emptyset $ in the
corresponding modified dag-like formalism, as follows.
\end{remark}

\section{Modified dag-like calculus \textsc{NM}$_{\rightarrow }^{\mathrm{%
\flat }}$}

As mentioned above, our modified dag-like deduction calculus, \textsc{NM}$%
_{\rightarrow }^{\mathrm{\flat }}$, includes inference rules $\left(
\rightarrow I\right) $, $\left( \rightarrow E\right) $, $\left( \rightarrow
R\right) $, $\left( \rightarrow S\right) $ (see Introduction). $\left(
\rightarrow I\right) $, $\left( \rightarrow R\right) $ and $\left(
\rightarrow E\right) $ have one and two premises, respectively, whereas $%
\left( \rightarrow S\right) $ has two or more ones. \textsc{NM}$%
_{\rightarrow }^{\mathrm{\flat }}$-deductions are graphically interpreted as
labeled rooted regular dags (abbr.: \emph{redags}, cf. \cite{GH}) $\partial
=\,\left\langle \text{\textsc{v}}\left( \partial \right) ,\text{\textsc{e}}%
\left( \partial \right) \right\rangle \,$, whose nodes may have arbitrary
many parents (conclusions) -- and children (premises), just in the case $%
\left( \rightarrow S\right) $, -- if any at all. The nodes ($x$, $y$, $z$,
...) are labeled by $\ell ^{\text{\textsc{f}}}$ with purely implicational
formulas ($\alpha $, $\beta $, $\gamma $, $\rho $, ...). For the sake of
brevity we'll assume that nodes $x$ are supplied with auxiliary\ \emph{%
height numbers} $h\left( x\right) \in \mathbb{N}$, while all inner nodes
also have special labels $\ell ^{\text{\textsc{n}}}\left( x\right) \in
\left\{ \text{\textsc{i}}_{\rightarrow }\text{\textsc{,e}}_{\rightarrow }%
\text{\textsc{,r}}_{\rightarrow }\text{\textsc{,s}}_{\rightarrow }\right\} $
showing the names of the inference rules $\left( \rightarrow I\right) $, $%
\left( \rightarrow E\right) $, $\left( \rightarrow R\right) $, $\left(
\rightarrow S\right) $ with conclusion $x$. The \emph{roots} and \emph{%
root-formulas} are always designated $r$ and $\rho :=\ell ^{\text{\textsc{f}}%
}\left( r\right) $, respectively. The \emph{edges} $\left\langle
x,y\right\rangle \in \,$\textsc{e}$\left( \partial \right) \subset \,$%
\textsc{v}$\left( \partial \right) ^{2}$ are directed upwards (thus $r$ is
the lowest node in $\partial $) in which $x$ and $y$ are called \emph{parents%
} and \emph{children} of each other, respectively. The \emph{leaves} \textsc{%
l}$\left( \partial \right) \subseteq \,$\textsc{v}$\left( \partial \right) $
are the nodes without children. Tree-like \textsc{NM}$_{\rightarrow }^{%
\mathrm{\flat }}$-deductions are those ones whose redags are trees (whose
nodes have at most one parent).

\begin{definition}
A given NM$_{\rightarrow }^{\mathrm{\flat }}$-deduction $\partial $ is \emph{%
locally correct} if conditions 1--2 are satisfied, for any $x,y,z,u\in \,$%
\textsc{v}$\left( \partial \right) $.

\begin{enumerate}
\item  $\partial $ is regular (cf. \cite{GH}), i.e.

\begin{enumerate}
\item  if $\left\langle x,y\right\rangle \in \,$\textsc{e}$\left( \partial
\right) $ then $x\notin \,$\textsc{l}$\left( \partial \right) $ and $y\neq r$%
,

\item  $h\left( r\right) =0$,

\item  if $\left\langle x,y\right\rangle ,\left\langle x,z\right\rangle \in
\,$\textsc{e}$\left( \partial \right) $ then $h\left( y\right) =h\left(
z\right) =h\left( x\right) +1$.
\end{enumerate}

\item  $\partial $ formalizes the inference rules, i.e. \ 

\begin{enumerate}
\item  if $\ell ^{\text{\textsc{n}}}\left( x\right) =\,$\textsc{r}$%
_{\rightarrow }$ and $\left\langle x,y\right\rangle ,\left\langle
x,z\right\rangle \in \,$\textsc{e}$\left( \partial \right) $ then $y=z$ and $%
\ell ^{\text{\textsc{f}}}\left( y\right) =\ell ^{\text{\textsc{f}}}\left(
x\right) $ [: rule $\left( \rightarrow R\right) $],

\item  if $\ell ^{\text{\textsc{n}}}\left( x\right) =\,$\textsc{i}$%
_{\rightarrow }$ and $\left\langle x,y\right\rangle ,\left\langle
x,z\right\rangle \in \,$\textsc{e}$\left( \partial \right) $ then $y=z$ and $%
\ell ^{\text{\textsc{f}}}\left( x\right) =\alpha \rightarrow \ell ^{\text{%
\textsc{f}}}\left( y\right) $ for some (uniquely determined) $\alpha $ [:
rule $\left( \rightarrow I\right) $],

\item  if $\ell ^{\text{\textsc{n}}}\left( x\right) =\,$\textsc{e}$%
_{\rightarrow }$ and $\left\langle x,y\right\rangle ,\left\langle
x,z\right\rangle ,\left\langle x,u\right\rangle \in \,$\textsc{e}$\left(
\partial \right) $ then $\left| \left\{ x,y,z\right\} \right| =2$ and if $%
y\neq z$ then either $\ell ^{\text{\textsc{f}}}\left( z\right) =\ell ^{\text{%
\textsc{f}}}\left( y\right) \rightarrow \ell ^{\text{\textsc{f}}}\left(
x\right) $ or else $\ell ^{\text{\textsc{f}}}\left( y\right) =\ell ^{\text{%
\textsc{f}}}\left( z\right) \rightarrow \ell ^{\text{\textsc{f}}}\left(
x\right) $ [: rule $\left( \rightarrow E\right) $],

\item  if $\ell ^{\text{\textsc{n}}}\left( x\right) =\,$\textsc{s}$%
_{\rightarrow }$ and $\left\langle x,y\right\rangle \in \,$\textsc{e}$\left(
\partial \right) $ then $\ell ^{\text{\textsc{f}}}\left( y\right) =\ell ^{%
\text{\textsc{f}}}\left( x\right) $ and $\ell ^{\text{\textsc{n}}}\left(
y\right) \neq \,$\textsc{s}$_{\rightarrow }$ [: rule $\left( \rightarrow
S\right) $].
\end{enumerate}
\end{enumerate}
\end{definition}

\textsc{NM}$_{\rightarrow }^{\ast }$ is easily embeddable into \textsc{NM}$%
_{\rightarrow }^{\mathrm{\flat }}$. Namely, consider a locally correct 
\textsc{NM}$_{\rightarrow }^{\ast }$-deduction frame $\widetilde{D}%
=\!\left\langle D,\text{\textsc{s}},\ell ^{\text{\textsc{f}}}\right\rangle $%
. The corresponding locally correct dag-like \textsc{NM}$_{\rightarrow }^{%
\mathrm{\flat }}$-deduction $\partial $ arises from $D$\ by ascending
recursion on the height.\thinspace \footnote{{\footnotesize For brevity we
omit }$h${\footnotesize , as every }$h\left( x\right) ${\footnotesize \ is
uniquely determined by }$x${\footnotesize .}} The root and\ basic
configurations of types $\left( \rightarrow I\right) $, $\left( \rightarrow
E\right) $, $\left( \rightarrow R\right) $ in $\widetilde{D}$\ should remain
unchanged. Furthermore, if $u$ is a node having several groups of premises
in $D$, i.e. $\left| \text{\textsc{s}}\left( u,D\right) \right| >1$ (cf. 
\cite{GH}) then in $\partial $ we separate these groups via $\left(
\rightarrow S\right) $ with $\left| \text{\textsc{s}}\left( u,D\right)
\right| $ identical premises; thus for example \textsc{NM}$_{\rightarrow
}^{\ast }$-configuration in $\widetilde{D}$ 
\begin{equation*}
\fbox{$\ \dfrac{\dfrac{\beta \quad \quad \gamma \quad \quad \gamma
\rightarrow \left( \alpha \rightarrow \beta \right) }{\alpha \rightarrow
\beta }}{\gamma \rightarrow \left( \alpha \rightarrow \beta \right) }$}
\end{equation*}
goes to this \textsc{NM}$_{\rightarrow }^{\mathrm{\flat }}$ -configuration
in $\partial $ 
\begin{equation*}
\fbox{$\left( \rightarrow I\right) \ \dfrac{\left( \rightarrow S\right) \ 
\dfrac{\ \left( \rightarrow I\right) \ \dfrac{\beta }{\alpha \rightarrow
\beta \ }\quad \left( \rightarrow E\right) \ \dfrac{\gamma \quad \quad
\gamma \rightarrow \left( \alpha \rightarrow \beta \right) }{\alpha
\rightarrow \beta \ }}{\alpha \rightarrow \beta \ }}{\gamma \rightarrow
\left( \alpha \rightarrow \beta \right) \ }$}\text{.}
\end{equation*}
Corresponding $\ell ^{\text{\textsc{f}}}$- and $\ell ^{\text{\textsc{n}}}$%
-labels are induced in an obvious way. Note that the weight of $\partial $
is linear in that of $\widetilde{D}$. \footnote{{\footnotesize Recall that
according to \cite{GH} we can just as well assume that }$\widetilde{D}$%
{\footnotesize \ is horizontally compressed and its weight is polynomial in
that of }$\rho $.}

Now suppose that there is a \emph{fst} $\mathcal{F}$ in a chosen \textsc{NM}$%
_{\rightarrow }^{\ast }$-deduction frame $\widetilde{D}$, and let $\mathcal{F%
}^{\mathrm{\flat }}$ be the image of $\mathcal{F}$ in $\partial $. It is
readily seen that $\mathcal{F}^{\mathrm{\flat }}$ is also a dense and $%
\left( \rightarrow E\right) $ preserving set of closed threads in $\partial $
(see \textsc{NM}$_{\rightarrow }^{\ast }$-clauses 1--3 in Ch. 1.2). That is, 
$\mathcal{F}^{\mathrm{\flat }}$ is a dense set of closed threads in $%
\partial $ such that for every $\Theta \in \mathcal{F}^{\mathrm{\flat }}$
and $\left( \rightarrow E\right) $-conclusion $x\in \Theta $, $\ell ^{\text{%
\textsc{n}}}\left( x\right) =\,$\textsc{e}$_{\rightarrow }$, with premises $%
y $ and $z$, if $y\in \Theta $ then there is a $\Theta ^{\prime }\in 
\mathcal{F}^{\mathrm{\flat }}$ such that $z\in \Theta ^{\prime }$ and $%
\Theta $ coincides with $\Theta ^{\prime }$ below $x$.

\subsection{Modified dag-like provability}

We formalize in \textsc{NM}$_{\rightarrow }^{\mathrm{\flat }}$ the modified
assignment $\mathcal{A}:\partial \ni x\hookrightarrow A\left( x\right)
\subseteq \mathrm{FOR}\left( \partial ^{\flat }\right) $.

\begin{definition}[Assignment $\mathcal{A}$]
Let $\partial ^{\mathrm{\flat }}$\ be any locally correct dag-like \textsc{NM%
}$_{\rightarrow }^{\mathrm{\flat }}$-deduction. Assign nodes $x\in \partial $
with terms $A\left( x\right) $ by descending recursion 1--5.

\begin{enumerate}
\item  $A\left( x\right) :=\left\{ \alpha \right\} $ if $x$ is a leaf and $%
\ell ^{\text{\textsc{f}}}\left( x\right) =$ $\alpha $.

\item  $A\left( x\right) :=A\left( y\right) $ if $\ell ^{\text{\textsc{n}}%
}\left( x\right) =\text{\textsc{r}}_{\rightarrow }$ and $\left\langle
x,y\right\rangle \in \,$\textsc{e}$\left( \partial \right) $.

\item  $A\left( x\right) :=A\left( y\right) \setminus \left\{ \alpha
\right\} $ if $\ell ^{\text{\textsc{n}}}\left( x\right) =\text{\textsc{i}}%
_{\rightarrow }$, $\left\langle x,y\right\rangle \in \,$\textsc{e}$\left(
\partial \right) $ and $\ell ^{\text{\textsc{f}}}\left( x\right) =\alpha
\rightarrow \ell ^{\text{\textsc{f}}}\left( y\right) $.

\item  $A\left( x\right) :=A\left( y\right) \cup A\left( z\right) $ if $\ell
^{\text{\textsc{n}}}\left( x\right) =\text{\textsc{e}}_{\rightarrow }$ and $%
\left\langle x,y\right\rangle ,\left\langle x,z\right\rangle \in \,\text{%
\textsc{e}}\left( \partial \right) $.

\item  $A\left( x\right) :=\circledS \left( A\!\left( y_{1}\!\right) ,\cdots
,\,A\!\left( y_{n}\!\right) \right) $ if $\ell ^{\text{\textsc{n}}}\left(
x\right) \!=\!\text{\textsc{s}}_{\rightarrow }$ \negthinspace and $\left(
\forall i\!\in \!\left[ 1,n\right] \right) \left\langle x,y_{i}\right\rangle
\!\in \,$\textsc{e}$\left( \partial \right) $.
\end{enumerate}
\end{definition}

\begin{definition}[Nondeterministic reduction]
Let $\partial $\ and $\mathcal{A}$\ be as above, $r$\ the root of $\partial $%
, $S$\ a set of formulas occurring in $\partial $. We say that $A\left(
r\right) $\emph{\ reduces} to $S$ (abbr.: $A\left( r\right)
\!\vartriangleright $ $\!S$) if $S$ arises from $A\left( r\right) $ by
successive (in a left-to-right direction) substitutions $A\left( u\right)
=\circledS \left( A\left( v_{1}\right) ,\cdots ,\,A\left( v_{n}\right)
\right) :=A\left( v_{i}\right) $, for a fixed chosen $i\in \left\{ 1,\cdots
,n\right\} $ and for any occurrence $A\left( u\right) $ in $A\left( w\right) 
$ and hence in $A\left( w^{\prime }\right) $, for every $w^{\prime }$ below $%
w$, provided that $u$ is a premise of $w$ such that $\ell ^{\text{\textsc{n}}%
}\left( u\right) =\text{\textsc{s}}_{\rightarrow }$,\footnote{{\footnotesize %
This operation is graphically interpreted by deleting }$u${\footnotesize \
along with }$v_{j}${\footnotesize \ for all }$j\neq i${\footnotesize . } 
{\footnotesize \ }} while using ordinary set-theoretic interpretations of $%
\cup $ and $\setminus $. We call $\partial $ a \emph{modified dag-like proof}
of $\rho =\ell ^{\text{\textsc{f}}}\left( r\right) $ (abbr.: $\partial
\vdash \rho $) if $A\left( r\right) \vartriangleright \emptyset $ holds.
\end{definition}

\begin{example}
Previously shown configuration yields a $\partial $ such that $\partial
\nvdash \rho $ : 
\begin{equation*}
\fbox{$\dfrac{\dfrac{\dfrac{\beta \,;A=\left\{ \beta \right\} }{\alpha
\rightarrow \beta :\text{\textsc{i}}_{\rightarrow }\,;A=\left\{ \beta
\right\} \ }\quad \dfrac{\gamma \,;A=\left\{ \gamma \right\} \quad \quad
\gamma \rightarrow \left( \alpha \rightarrow \beta \right) ;A=\left\{ \gamma
\rightarrow \left( \alpha \rightarrow \beta \right) \right\} }{\alpha
\rightarrow \beta :\text{\textsc{e}}_{\rightarrow }\,;A=\left\{ \gamma
,\gamma \rightarrow \left( \alpha \rightarrow \beta \right) \right\} }}{%
\alpha \rightarrow \beta \,:\text{\textsc{s}}_{\rightarrow };A=\circledS
\left( \left\{ \beta \right\} ,\left\{ \gamma ,\gamma \rightarrow \left(
\alpha \rightarrow \beta \right) \right\} \right) }}{\gamma \rightarrow
\left( \alpha \rightarrow \beta \right) :\text{\textsc{i}}_{\rightarrow
}\,;A=\circledS \left( \left\{ \beta \right\} ,\left\{ \gamma \rightarrow
\left( \alpha \rightarrow \beta \right) \right\} \right) \ }$}
\end{equation*}
where $\ell ^{\text{\textsc{n}}}\left( r\right) =\,$\textsc{i}$_{\rightarrow
}$, $\ell ^{\text{\textsc{f}}}\left( r\right) \!=\!\rho \!=\!\gamma
\rightarrow \left( \alpha \rightarrow \beta \right) $ and $A\left( r\right)
\!=\!\circledS \left( \left\{ \beta \right\} ,\left\{ \gamma \rightarrow
\left( \alpha \rightarrow \beta \right) \right\} \right) $. Note that $%
A\left( r\right) \vartriangleright \left\{ \beta \right\} $ and $A\left(
r\right) \vartriangleright \left\{ \gamma \rightarrow \left( \alpha
\rightarrow \beta \right) \right\} $, although $A\left( r\right)
\ntriangleright \emptyset $.

To obtain an analogous dag-like proof of (say) $\rho ^{\prime }:=\beta
\rightarrow \left( \gamma \rightarrow \left( \alpha \rightarrow \beta
\right) \right) $ we'll upgrade $\partial $ to such $\partial ^{\prime }$\ : 
\begin{equation*}
\fbox{$\dfrac{\dfrac{\dfrac{\dfrac{\beta \,;A=\left\{ \beta \right\} }{%
\alpha \rightarrow \beta :\text{\textsc{i}}_{\rightarrow }\,;A=\left\{ \beta
\right\} \ }\quad \dfrac{\gamma \,;A=\left\{ \gamma \right\} \quad \quad
\gamma \rightarrow \left( \alpha \rightarrow \beta \right) ;A=\left\{ \gamma
\rightarrow \left( \alpha \rightarrow \beta \right) \right\} }{\alpha
\rightarrow \beta :\text{\textsc{e}}_{\rightarrow }\,;A=\left\{ \gamma
,\gamma \rightarrow \left( \alpha \rightarrow \beta \right) \right\} }}{%
\alpha \rightarrow \beta \,:\text{\textsc{s}}_{\rightarrow };A=\circledS
\left( \left\{ \beta \right\} ,\left\{ \gamma ,\gamma \rightarrow \left(
\alpha \rightarrow \beta \right) \right\} \right) }}{\gamma \rightarrow
\left( \alpha \rightarrow \beta \right) :\text{\textsc{i}}_{\rightarrow
}\,;A=\circledS \left( \left\{ \beta \right\} ,\left\{ \gamma \rightarrow
\left( \alpha \rightarrow \beta \right) \right\} \right) \setminus \left\{
\gamma \right\} \ }}{\beta \rightarrow \left( \gamma \rightarrow \left(
\alpha \rightarrow \beta \right) \right) :\text{\textsc{i}}_{\rightarrow
}\,;A=\circledS \left( \left\{ \beta \right\} ,\left\{ \gamma \rightarrow
\left( \alpha \rightarrow \beta \right) \right\} \right) \setminus \left\{
\gamma \right\} \setminus \left\{ \beta \right\} \ }$}
\end{equation*}
and let $\circledS \left( \left\{ \beta \right\} ,\left\{ \gamma ,\gamma
\rightarrow \left( \alpha \rightarrow \beta \right) \right\} \right)
:=\left\{ \beta \right\} $. Then $A\left( r\right) \vartriangleright
\emptyset $ , i.e. $\partial ^{\prime }\vdash $ $\rho _{1}$ holds.
\end{example}

\begin{lemma}
Every modified dag-like proof of $\rho $ is convertible to a dag-like 
\textsc{NM}$_{\rightarrow }$-proof of $\rho $.
\end{lemma}

\begin{proof}
Let $\partial $ be a given \textsc{NM}$_{\rightarrow }^{\mathrm{\flat }}$%
-proof of $\rho $. Its \textsc{NM}$_{\rightarrow }$-conversion is defined by
a simple ascending recursion, as follows. Each time we arrive at a $w$ whose
premise $u$ is a conclusion of $\left( \rightarrow S\right) $, we replace
this $u$\ by its chosen premise that is ``guessed'' by a given
nondeterministic reduction leading to $A\left( r\right) \vartriangleright
\emptyset $ -- alternatively, we replace this $\left( \rightarrow S\right) $
by the corresponding repetition $\left( \rightarrow R\right) $. It is
readily seen that the resulting dag-like deduction $\partial _{0}$ with the
same root-formula $\rho $\ is locally correct and $\left( \rightarrow
S\right) $-free, and hence it belongs to \textsc{NM}$_{\rightarrow }$.
Obviously this conversion preserves a given assignment $x\hookrightarrow
A\left( x\right) $. Also note that $A\left( r\right) \vartriangleright
\emptyset $ in $\partial $ infers $A\left( r\right) =\emptyset $ in $%
\partial _{0}$, and hence $\partial _{0}$ proves $\rho $ in \textsc{NM}$%
_{\rightarrow }$.
\end{proof}

\begin{lemma}
Let $\widetilde{D}$ be any locally correct deduction frame in \textsc{NM}$%
_{\rightarrow }^{\ast }$ with root-formula $\rho $ that admits some fst.
There exists a dag-like \textsc{NM}$_{\rightarrow }$-proof of $\rho $\ whose
weight does not exceed that of $\widetilde{D}$.
\end{lemma}

\begin{proof}
Let $\partial $ be the \textsc{NM}$_{\rightarrow }^{\mathrm{\flat }}$%
-deduction of $\rho $ induced by $\widetilde{D}$\ and $\mathcal{F}$ any 
\emph{fst} in $\widetilde{D}$. Furthermore, let $\mathcal{F}^{\mathrm{\flat }%
}$ be the image of $\mathcal{F}$ in $\partial $ (see above). We'll show that 
$\mathcal{F}^{\mathrm{\flat }}$ determines successive left-to-right $%
\circledS $-eliminations $\circledS \left( \!A\left( y_{1}\right) ,\cdots
,\,A\left( y_{n}\right) \!\right) \!\hookrightarrow \!A\left( y_{i}\right) $
inside $A\left( r\right) $ leading to required reduction $A\left( r\right)
\vartriangleright \emptyset $. These eliminations together with a suitable
sub-\emph{fst} $\mathcal{F}_{0}^{\mathrm{\flat }}$ $\subseteq \mathcal{F}^{%
\mathrm{\flat }}$ arise as follows\ by ascending recursion along $\mathcal{F}%
^{\mathrm{\flat }}$. Let $x$ with $\ell ^{\text{\textsc{n}}}\left( x\right) =%
\text{\textsc{e}}_{\rightarrow }$ be a chosen lowest conclusion of $\left(
\rightarrow E\right) $ in $\partial $, if any exists. By the density of $%
\mathcal{F}^{\mathrm{\flat }}$, there exists $\Theta \in \mathcal{F}^{%
\mathrm{\flat }}$ with $x\in \Theta $; we let $\Theta \in \mathcal{F}_{0}^{%
\mathrm{\flat }}$. Let $y$ and $z$\ be the two premises of $x$ and suppose
that $y\in \Theta $. By the third \emph{fst} condition there exists a $%
\Theta ^{\prime }\in \mathcal{F}^{\mathrm{\flat }}$ with $z\in \Theta
^{\prime }$ and $\Theta \!\upharpoonright _{x}=$ $\Theta ^{\prime
}\!\upharpoonright _{x}$; so let $\Theta ^{\prime }\in \mathcal{F}_{0}^{%
\mathrm{\flat }}$ be the corresponding ``upgrade''of $\Theta $. In the case $%
z\in \Theta $ we let $\Theta ^{\prime }:=\Theta $. Note that $\Theta
\!\!\upharpoonright _{x}$ determines substitutions $A\left( u\right)
=\circledS \left( A\left( v_{1}\right) ,\cdots ,\,A\left( v_{n}\right)
\right) :=A\left( v_{i}\right) $ in all parents of $\left( \rightarrow
S\right) $-conclusions $u$ occurring in both $\Theta $ and $\Theta ^{\prime
} $\ below $x$ (cf. Definitions 10, 11), if any exist, and hence also $%
\circledS $-eliminations $A\left( u\right) \hookrightarrow A\left(
v_{i}\right) $ in the corresponding subterms of $A\left( r\right) $. The
same procedure is applied to the nodes occurring in $\Theta $ and $\Theta
^{\prime }$\ above $x$ under the next lowest conclusions of $\left(
\rightarrow E\right) $; this yields new ``upgraded'' threads $\Theta
^{\prime \prime },\Theta ^{\prime \prime \prime },\cdots \in \mathcal{F}%
_{0}^{\mathrm{\flat }}$ and $\circledS $-eliminations in the corresponding
initial fragments of $A\left( r\right) $. We keep doing this recursively
until the list of remaining $\circledS $-occurrences in $\Theta \in \mathcal{%
F}_{0}^{\mathrm{\flat }}$ is empty. The final ``cleansed'' $\circledS $-free
form of $A\left( r\right) $ is represented by a set of formulas that easily
reduces to $\emptyset $ by ordinary set-theoretic interpretation of the
remaining operations $\cup $ and $\setminus $, since every $\Theta \in 
\mathcal{F}_{0}^{\mathrm{\flat }}$ involved is closed. That is, the
correlated ``cleansed'' deduction $\partial _{0}$ is a locally correct
dag-like deduction of $\rho $ in the $\left( \rightarrow S\right) $-free
fragment of \textsc{NM}$_{\rightarrow }^{\mathrm{\flat }}$, and hence it
belongs to \textsc{NM}$_{\rightarrow }$; moreover the set of ascending
threads in $\partial _{0}$ is uniquely determined by the remaining rules $%
\left( \rightarrow R\right) $, $\left( \rightarrow I\right) $, $\left(
\rightarrow E\right) $ (cf. analogous passage in the previous proof). Now by
the definition these ``cleansed'' ascending threads are all included in $%
\mathcal{F}_{0}^{\mathrm{\flat }}$ and hence closed with respect to $\left(
\rightarrow I\right) $. \footnote{{\footnotesize These threads may be
exponential in number, but our nondeterministic algorithm runs on the
polynomial set of nodes. }} This yields a desired reduction $A\left(
r\right) \vartriangleright \emptyset $. Hence $\partial _{0}$ proves $\rho $
in \textsc{NM}$_{\rightarrow }$. Obviously the weight of $\partial _{0}$
does not exceed the weight of $\widetilde{D}$.
\end{proof}

Operation $\partial \hookrightarrow \partial _{0}$ is referred to as \emph{%
horizontal cleansing} (cf. Introduction). Together with Remark 2 and
Corollary 7 these lemmata yield

\begin{corollary}
Any given $\rho $ is valid in minimal logic iff it is provable in \textsc{NM}%
$_{\rightarrow }$ by a dag-like deduction $\partial _{0}$ whose weight is
polynomial in that of $\rho $ and such that $\partial _{0}\vdash \rho $ can
be confirmed by a deterministic TM in $\left| \rho \right| $-polynomial
time. \footnote{{\footnotesize See Appendix for a more exhaustive
presentation.}}
\end{corollary}

\begin{theorem}
$\mathbf{PSPACE}\subseteq \mathbf{NP}$ and hence $\mathbf{NP=PSPACE}$.
\end{theorem}

\begin{proof}
Minimal propositional logic is PSPACE-complete (cf. e.g. \cite{Joh}, \cite
{Statman}, \cite{Svejdar}). Hence $\mathbf{PSPACE}\subseteq \mathbf{NP}$
directly follows from Corollary 15. Note that in contrast to \cite{GH} here
we use nondeterministic arguments twice. First we ``guess'' the existence of
Hudelmaier-style cutfree sequential deduction of $\rho $\ that leads (by
deterministic compression) to a ``small'' natural deduction frame $%
\widetilde{D}$ that is supposed to have a \emph{fst} $\mathcal{F}$. Having
this we ``guess'' the existence of \ ``cleansed'' modified subdeduction that
confirms in $\left| \rho \right| $-polynomial time the provability of $\rho $
with regard to $\left\langle \widetilde{D},\mathcal{F}\right\rangle $.
\end{proof}

\begin{corollary}
$\mathbf{NP=Co}$\textbf{-}$\mathbf{NP}$ and hence the polynomial hierarchy
collapses to the first level.
\end{corollary}

\begin{proof}
$\mathbf{NP\!=\!PSPACE}$ implies $\mathbf{Co}$\textbf{-}$\mathbf{NP\!=\!Co}$%
\textbf{-}$\mathbf{PSPACE\!=\!PSPACE\!=\!NP}$ (see also \cite{Papa}).
\end{proof}

\section{Appendix: rough complexity estimate}

\subsection{Dag-like proof system \textsc{NM}$_{\rightarrow }$}

We regard \textsc{NM}$_{\rightarrow }$\ as \textsc{NM}$_{\rightarrow
}^{\flat }$ without separation rule $\left( \rightarrow S\right) $.
Moreover, without loss of generality we suppose that dag-like \textsc{NM}$%
_{\rightarrow }$-deductions $\partial $ of root-formulas $\rho $ have
polynomial total number of vertices $\left| \text{\textsc{v}}\left( \partial
\right) \right| =\mathcal{O}\left( \left| \rho \right| ^{4}\right) $ while
the lengths (weights) of formulas and the height numbers involved are
bounded by $2\left| \rho \right| $ and $\left| \text{\textsc{v}}\left(
\partial \right) \right| $, respectively (cf. \cite{GH}).

Let $\mathrm{LC}\left( \partial \right) $ and $\mathrm{PROV}\left( \partial
\right) $\ be abbreviations for `$\partial $\emph{\ is locally correct}' and
`$\partial $ \emph{proves} $\rho $', respectively, and let $\mathrm{PROOF}%
\left( \partial \right) :=\mathrm{LC}\left( \partial \right) \&\,\mathrm{PROV%
}\left( \partial \right) $. We wish to validate the assertion $\mathrm{PROOF}%
\left( \partial \right) $ in polynomial time (and space) by a suitable
deterministic TM $M$. For technical reasons we choose a formalization of $%
\partial $ in which edges are redefined as pairs $\left\langle
parent,child\right\rangle $.

\textbf{Input:} $a:=2\left| \rho \right| $, $b:=\left| \text{\textsc{v}}%
\left( \partial \right) \right| $ and the list $t$ consisting of tuples $%
t\left( x\right) =\left[ x,y_{1},y_{2},h,h_{1},h_{2},\chi ,\gamma ,\beta
_{1},\beta _{2}\right] $, for every $x\leq b$, where $\chi \in \left\{ \text{%
\textsc{r}}_{\rightarrow },\text{\textsc{i}}_{\rightarrow },\text{\textsc{e}}%
_{\rightarrow },\text{\textsc{l}}\right\} $ and $h,h_{1},h_{2}\leq b$ are
natural numbers (the heights, in binary) while $x,y_{1},y_{2}\leq b$ and $%
\gamma ,\beta _{1},\beta _{2}\leq a$ are natural numbers (in binary) which
encode nodes and formulas, respectively, occurring in $\partial $. Thus the
total length (weight) of $t$ is $\mathcal{O}\left( \left| \rho \right|
^{4}\log \left| \rho \right| \right) <\mathcal{O}\left( \left| \rho \right|
^{5}\right) $. $\mathrm{LC}\left( \partial \right) $ and $\mathrm{PROV}%
\left( \partial \right) $ are formalized as follows.

\subsection{Local correctness}

We observe that $\mathrm{LC}\left( \partial \right) $ is equivalent to
conjunction of the following conditions 1--8 on $t\left( x\right) $\ under
the assumptions: `$x$ \emph{is parent of} $y$', $h:=h\left( x\right) $, $%
h_{1}:=h\left( y_{1}\right) $, $h_{2}:=h\left( y_{2}\right) $, $\chi :=\ell
^{\text{\textsc{n}}}\left( x\right) $, $\gamma :=\ell ^{\text{\textsc{f}}%
}\left( x\right) $, $\beta _{1}:=\ell ^{\text{\textsc{f}}}\left(
y_{1}\right) $ and $\beta _{2}:=\ell ^{\text{\textsc{f}}}\left( y_{2}\right) 
$.

\begin{enumerate}
\item  If $x=x^{\prime }$\ then $t\left( x\right) =t\left( x^{\prime
}\right) $.

\item  If $t\left( x\right) =\left[ x,y_{1},y_{2},h,h_{1},h_{2},\chi ,\gamma
,\beta _{1},\beta _{2}\right] $ and $x^{\prime }=y_{i}$ ($i\in \left\{
1,2\right\} $) with $t\left( x^{\prime }\right) =\left[ x^{\prime
},y_{1}^{\prime },y_{2}^{\prime },h^{\prime },h_{1}^{\prime },h_{2}^{\prime
},\chi ^{\prime },\gamma ^{\prime },\beta _{1}^{\prime },\beta _{2}^{\prime }%
\right] $, then $h^{\prime }=h_{i}$ and $\gamma ^{\prime }=\beta _{i}$.

\item  If $x=r$ then $h=0$ and $\chi \neq $\thinspace \textsc{l.}

\item  If $\chi =\,$\textsc{l} then $y_{1}=y_{2}=h_{1}=h_{2}=\beta
_{1}=\beta _{2}=0$ [: case $x\in \,$\textsc{l}$\left( \partial \right) $].

\item  If $\chi \neq $\thinspace \textsc{l} then\textsc{\ }$h_{1}=h_{2}=h+1$%
\textsc{.}

\item  If $\chi =\,$\textsc{r}$_{\rightarrow }$ then $y_{2}=\beta _{2}=0$
and $\gamma =\beta _{1}$\textsc{\ }[: rule $\left( \rightarrow R\right) $].

\item  If $\chi =\,$\textsc{i}$_{\rightarrow }$ then $y_{2}=\beta _{2}=0$
and $\gamma =\alpha \rightarrow \beta _{1}$ for some $\alpha $ [: rule $%
\left( \rightarrow I\right) $].

\item  If $\chi =\,$\textsc{e}$_{\rightarrow }$ then $\beta _{2}=\beta
_{1}\rightarrow \gamma $ [: rule $\left( \rightarrow E\right) $].
\end{enumerate}

We assume that $\mathrm{LC}\left( \partial \right) $ is validated by a TM $M$
in polynomial\ time (and space). The verification of conditions 1--8
requires $\mathcal{O}\left( \left| \rho \right| ^{5}\right) $ iterations of
basic queries $\chi \overset{?}{=}\chi ^{\prime }$, $u\overset{?}{=}v$, $%
\delta \overset{?}{=}\sigma $, $\left( \exists ?\alpha \right) \gamma
=\alpha \rightarrow \beta $ for $\chi ,\chi ^{\prime }\in \left\{ \text{%
\textsc{r}}_{\rightarrow },\text{\textsc{i}}_{\rightarrow },\text{\textsc{e}}%
_{\rightarrow },\text{\textsc{l}}\right\} $, $u,v\leq b$ and $\beta ,\gamma
,\delta ,\sigma \leq a$ that are solvable in $\mathcal{O}\left( \left| \rho
\right| \right) $ time (note that $\alpha \rightarrow \beta =\,\rightarrow
\!\alpha \beta $ in the \L ukasiewicz prefix notation). Summing up there is
a deterministic TM $M$ that verifies $\mathrm{LC}\left( \partial \right) $
in $\mathcal{O}\left( \left| \rho \right| ^{5}\cdot \left| \rho \right|
\right) =\mathcal{O}\left( \left| \rho \right| ^{6}\right) $ time and $%
\mathcal{O}\left( \left| \rho \right| ^{5}\right) $ space.

\subsection{Assignment $\mathcal{A}$}

A given locally correct \textsc{NM}$_{\rightarrow }$-deduction $\partial $
determines an assignment 
\begin{equation*}
\mathcal{A}:\text{\textsc{v}}\left( \partial \right) \ni x\hookrightarrow
A\left( x\right) \subseteq \mathrm{FOR}\left( \partial \right)
\end{equation*}
that is defined by the following recursive clauses 1--4.

\begin{enumerate}
\item  $A\left( x\right) :=\left\{ \alpha \right\} $ if $x\in \,$\textsc{l}$%
\left( \partial \right) $ and $\ell ^{\text{\textsc{f}}}\left( x\right) =$ $%
\alpha $,

\item  $A\left( x\right) :=A\left( y\right) $ if $\ell ^{\text{\textsc{n}}%
}\left( x\right) =\,$\textsc{r}$_{\rightarrow }$ and $\left\langle
x,y\right\rangle \in \,$\textsc{e}$\left( \partial \right) $.

\item  $A\left( x\right) :=A\left( y\right) \setminus \left\{ \alpha
\right\} $ if $\ell ^{\text{\textsc{n}}}\left( x\right) =\text{\textsc{i}}%
_{\rightarrow }$, $\ell ^{\text{\textsc{f}}}\left( x\right) =\alpha
\rightarrow \ell ^{\text{\textsc{f}}}\left( y\right) $ and $\left\langle
x,y\right\rangle \in \,$\textsc{e}$\left( \partial \right) $.

\item  $A\left( x\right) :=A\left( y\right) \cup A\left( z\right) $ if $\ell
^{\text{\textsc{n}}}\left( x\right) =\text{\textsc{e}}_{\rightarrow }$ and $%
\left\langle x,y\right\rangle ,\left\langle x,z\right\rangle \in \,$\textsc{e%
}$\left( \partial \right) $.
\end{enumerate}

$\mathcal{A}$\ is defined by recursion of the length $\left| \text{\textsc{v}%
}\left( \partial \right) \right| =\mathcal{O}\left( \left| \rho \right|
^{4}\right) $. Recursion steps produce (say, sorted) lists of formulas $%
A\left( x\right) \subseteq \left\{ \ell ^{\text{\textsc{f}}}\left( y\right)
:y\in \text{\textsc{l}}\left( \partial \right) \right\} $, $x\in $\textsc{%
\thinspace v}$\left( \partial \right) $, $\left| A\left( x\right) \right|
\leq $ $\left| \text{\textsc{v}}\left( \partial \right) \right| $ using
set-theoretic unions $A\cup B$ and subtractions $A\setminus \left\{ \alpha
\right\} $. Each recursion step requires $\mathcal{O}\left( \left| \rho
\right| \cdot \left| \text{\textsc{v}}\left( \partial \right) \right|
\right) =\mathcal{O}\left( \left| \rho \right| ^{5}\right) $ steps of
computation. This yields upper bound\ $\mathcal{O}\left( \left| \rho \right|
^{4}\cdot \left| \rho \right| ^{5}\right) =\mathcal{O}\left( \left| \rho
\right| ^{9}\right) $ for $A\left( r\right) \overset{?}{=}\emptyset $. Thus $%
\mathrm{PROV}\left( \partial \right) $ is\ verifiable by a suitable
deterministic TM $M$\ in $\mathcal{O}\left( \left| \rho \right| ^{9}\right) $
time and $\mathcal{O}\left( \left| \rho \right| \right) $ space.\ Hence by
the above estimate of $\mathrm{LC}\left( \partial \right) $\ we can safely
assume that $\mathrm{PROOF}\left( \partial \right) $ is verifiable by $M$ in 
$\mathcal{O}\left( \left| \rho \right| ^{9}\right) $ time and $\mathcal{O}%
\left( \left| \rho \right| ^{5}\right) $ space.

\begin{conclusion}
There exist polynomials $p,q,r$ of degrees $5,9,5$, respectively, and a
deterministic boolean-valued TM $M$ such that for any purely implicational
formula $\rho $ the following holds: $\rho $\emph{\ is valid in minimal
logic iff there exists a }$u\in \left\{ 0,1\right\} ^{p\left( \left| \rho
\right| \right) }$\emph{\ such that }$M\left( \rho ,u\right) $\emph{\ yields 
}$1$\emph{\ after }$q\left( \left| \rho \right| +\left| u\right| \right) $%
\emph{\ steps of computation in space }$r\left( \left| \rho \right| +\left|
u\right| \right) $. Analogous polynomial estimates of the intuitionistic
and/or classical propositional and even quantified boolean validity are
easily obtained by familiar syntactic interpretations within minimal logic
(cf. e.g. \cite{Ish}, \cite{PraMa}, \cite{Svejdar}).
\end{conclusion}

\begin{remark}
Recall that $\mathrm{PROV}\left( \partial \right) \!$ is equivalent to
the\negthinspace\ assertion\negthinspace\ that\negthinspace\ maximal threads
in $\partial $ are closed. This in turn is equivalent to a variant of
non-reachability assertion: `$r$\emph{\ is not connected to any leaf }$z$%
\emph{\ in a subgraph of }$\partial $\emph{\ that is obtained by deleting
all edges }$\left\langle x,y\right\rangle $\emph{\ with }$\ell ^{\text{%
\textsc{n}}}\left( x\right) =\,$\textsc{i}$_{\rightarrow }$\emph{\ and }$%
\ell ^{\text{\textsc{f}}}\left( x\right) =\ell ^{\text{\textsc{f}}}\left(
z\right) \rightarrow \ell ^{\text{\textsc{f}}}\left( y\right) $', which
we'll abbreviate by $\mathrm{PROV}_{1}\!\left( \partial \right) $. Now $%
\mathrm{PROV}_{1}\!\left( \partial \right) $\ is verifiable by a
deterministic TM in $\mathcal{O}\left( \left| \text{\textsc{v}}\left(
\partial \right) \right| \cdot \left| \text{\textsc{e}}\left( \partial
\right) \right| \right) =\mathcal{O}\left( \left| \rho \right| ^{12}\right) $
time and $\mathcal{O}\left( \left| \rho \right| \cdot \left| \text{\textsc{v}%
}\left( \partial \right) \right| \right) =\mathcal{O}\left( \left| \rho
\right| ^{5}\right) $ space (cf. e.g. \cite{Papa}). However this does not
improve our upper bound for $\mathrm{PROOF}\left( \partial \right) $.
Actually there are known much better estimates of the reachability problem
(cf. e.g. \cite{Thor}, \cite{Holm}), but at this stage we are not interested
in more precise analysis.
\end{remark}

-----------------------------------------------------------------------------------------------

\end{document}